\documentclass[a4paper,10pt,twocolumn,preprint]{elsarticle}

 \usepackage{hyperref}
 \usepackage{amsmath}
 \usepackage{amssymb}
 \usepackage{tikz}
 \renewcommand{\dim}{\mathrm{dim}}
 \renewcommand{\ker}{\mathrm{ker}}
 \DeclareMathOperator{\Ima}{Im}
 \newcommand{\coker}{\mathrm{coker}}
 \DeclareMathOperator{\link}{link}
 \DeclareMathOperator{\sta}{star}

 \newdefinition{definition}{Definition}
 \newtheorem{remark}[definition]{Remark}
 \newtheorem{theorem}{Theorem}
 \newtheorem{lemma}[theorem]{Lemma}
 \newtheorem{corollary}[theorem]{Corollary}
 \newtheorem{proposition}[theorem]{Proposition}
 \newproof{proof}{Proof}

 \usepackage{geometry}
\begin{document}

\twocolumn[{
\begin{frontmatter}

\title{Critical sets of PL and discrete Morse theory: a correspondence}

\author[mymainaddress]{Ulderico Fugacci}

\author[mysecondaryaddress]{Claudia Landi}

\author[mythirdaddress]{Hanife Varl{\i}}

\address[mymainaddress]{Polytechnic University of Torino, Torino, Italy}
\address[mysecondaryaddress]{University of Modena and Reggio Emilia, Modena, Italy}
\address[mythirdaddress]{\c{C}ank{\i}r{\i} Karatekin University, \c{C}ank{\i}r{\i}, Turkey}

\begin{abstract}
Piecewise-linear (PL) Morse theory and discrete Morse theory are used in shape analysis tasks to investigate the topological features of discretized spaces.
In spite of their common origin in smooth Morse theory, various notions of critical points have been given in the literature for the discrete setting, making a clear understanding of the relationships occurring between them not obvious. This paper aims at providing equivalence results about critical points of the two discretized Morse theories.
First of all, we prove the equivalence of the existing notions of PL critical points. Next, under an optimality condition called relative perfectness, we show a dimension agnostic correspondence between the set of PL critical points and that of discrete critical simplices of the combinatorial approach. Finally, we show how a relatively perfect discrete gradient vector field can be algorithmically built up to dimension 3. This way, we guarantee a formal and operative connection between critical sets in the PL and discrete theories.
\end{abstract}



\end{frontmatter}
}]


\section{Introduction}
{\em Topological shape analysis} is useful to extract information about  topological and  morphological properties of a shape, naturally finding  applications in  fields that require shape understanding such as  computer graphics, computer vision and visualization. \cite{Biasotti-et-al2008} shows that most methods of topological shape analysis are grounded in {\em Morse theory} \cite{Milnor63}, as they investigate critical sets of  functions.

The proved effectiveness of Morse theory has led to the development of several discrete counterparts of this theory  useful when one works with shapes discretized as cell complexes, in particular  simplicial complexes \cite{DeFloriani15}. Among them, two discretized versions of Morse theory have gained a prominent role in the literature: the {\em piecewise-linear (PL) Morse theory} introduced by Banchoff \cite{Banchoff67} and the {\em discrete Morse theory} developed by Forman \cite{Forman98}.

Both the approaches are worth to be addressed as Morse theories on the ground that they satisfy discrete versions of the main theorem valid in the smooth case: topological changes  of a shape occur at the  critical points of a function defined on it (or, equivalently, at singularities of a gradient field).

In spite of these similarities, several aspects distinguish the two theories.
First of all, PL Morse theory is more centered on functions whereas discrete Morse theory is more based on gradient vector fields. Indeed,  a PL function is uniquely defined on each point of a polyhedron by linearly interpolation of  scalar values given at the vertices. In contrast, the  approach by Forman is combinatorial in that it treats simplices as a whole rather than as a set of points, and  can produce a collection of simplex pairs  to simulate the behaviour of a function gradient independently of whether or not a scalar function is globally defined.


\begin{figure*}[!htb]
	\centering
	\begin{tabular}{cccc}
    \includegraphics[width=.20\textwidth]{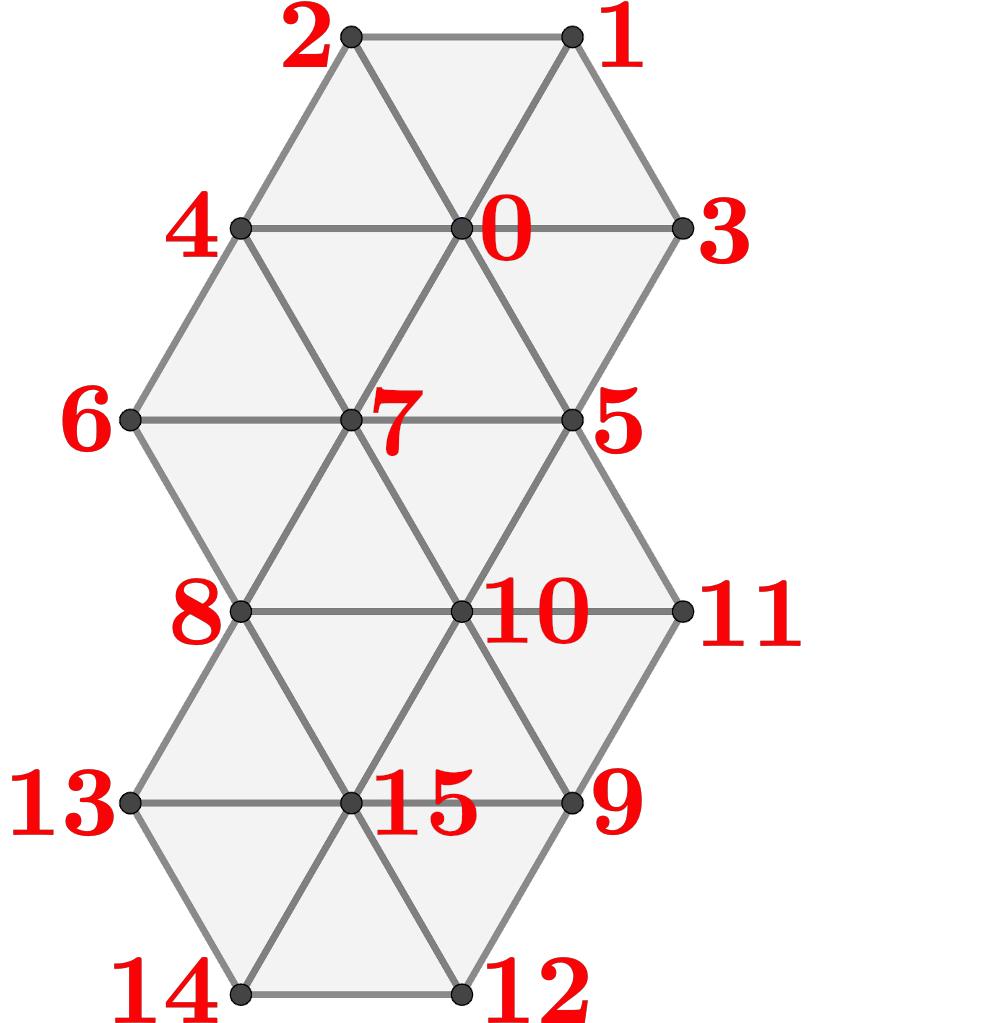} & \includegraphics[width=.24\textwidth]{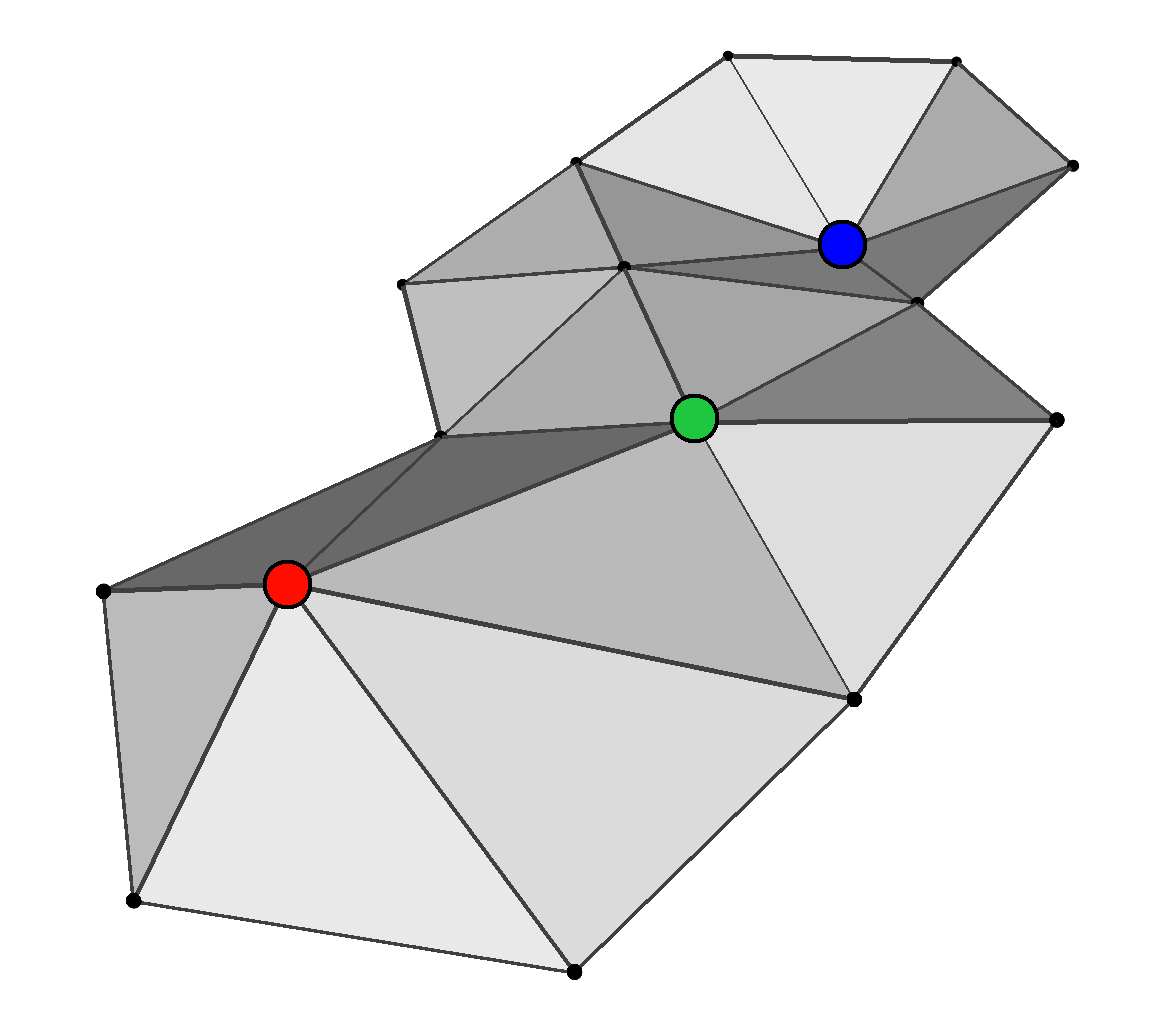} & \includegraphics[width=.24\textwidth]{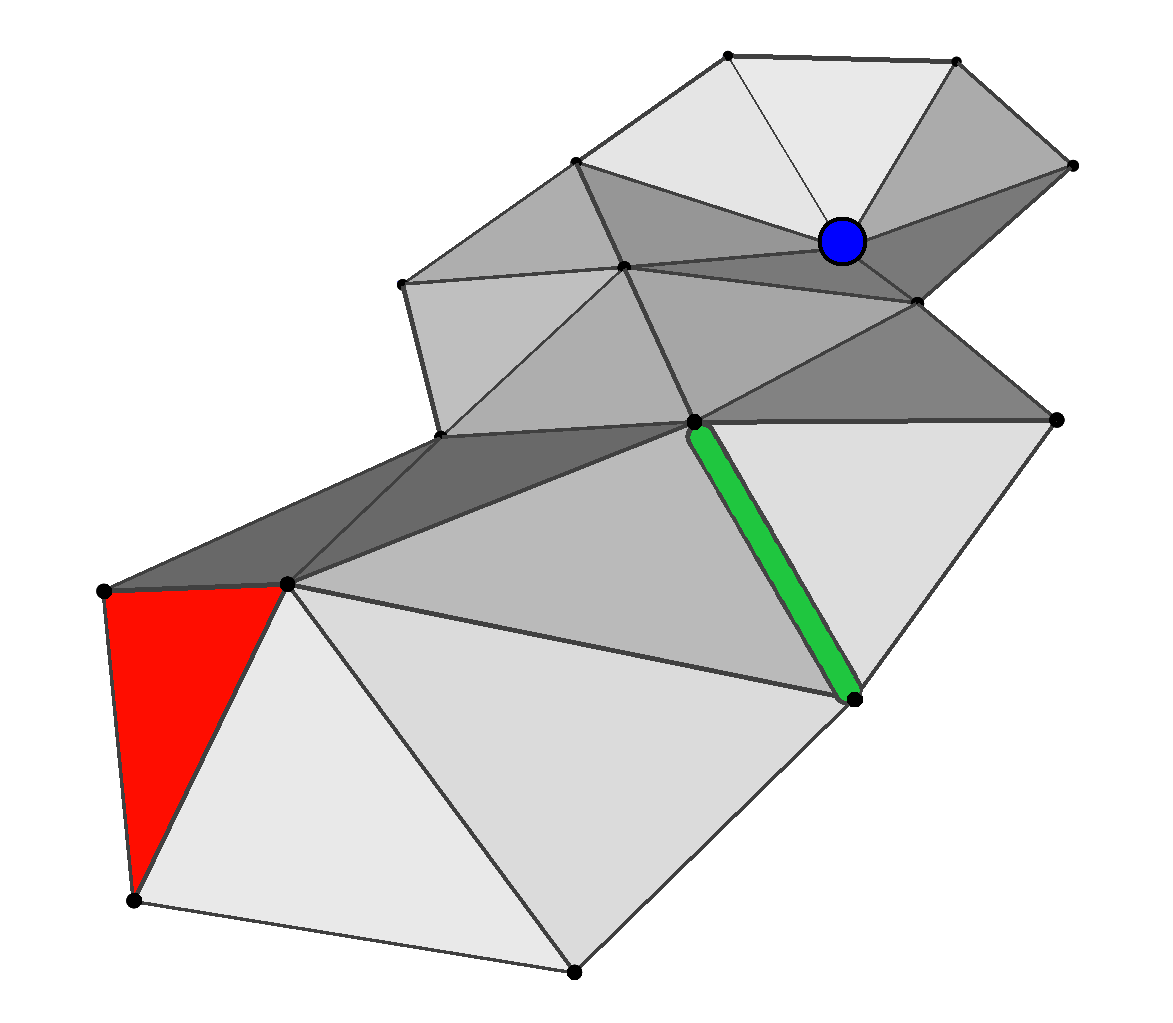} &
		\includegraphics[width=.24\textwidth]{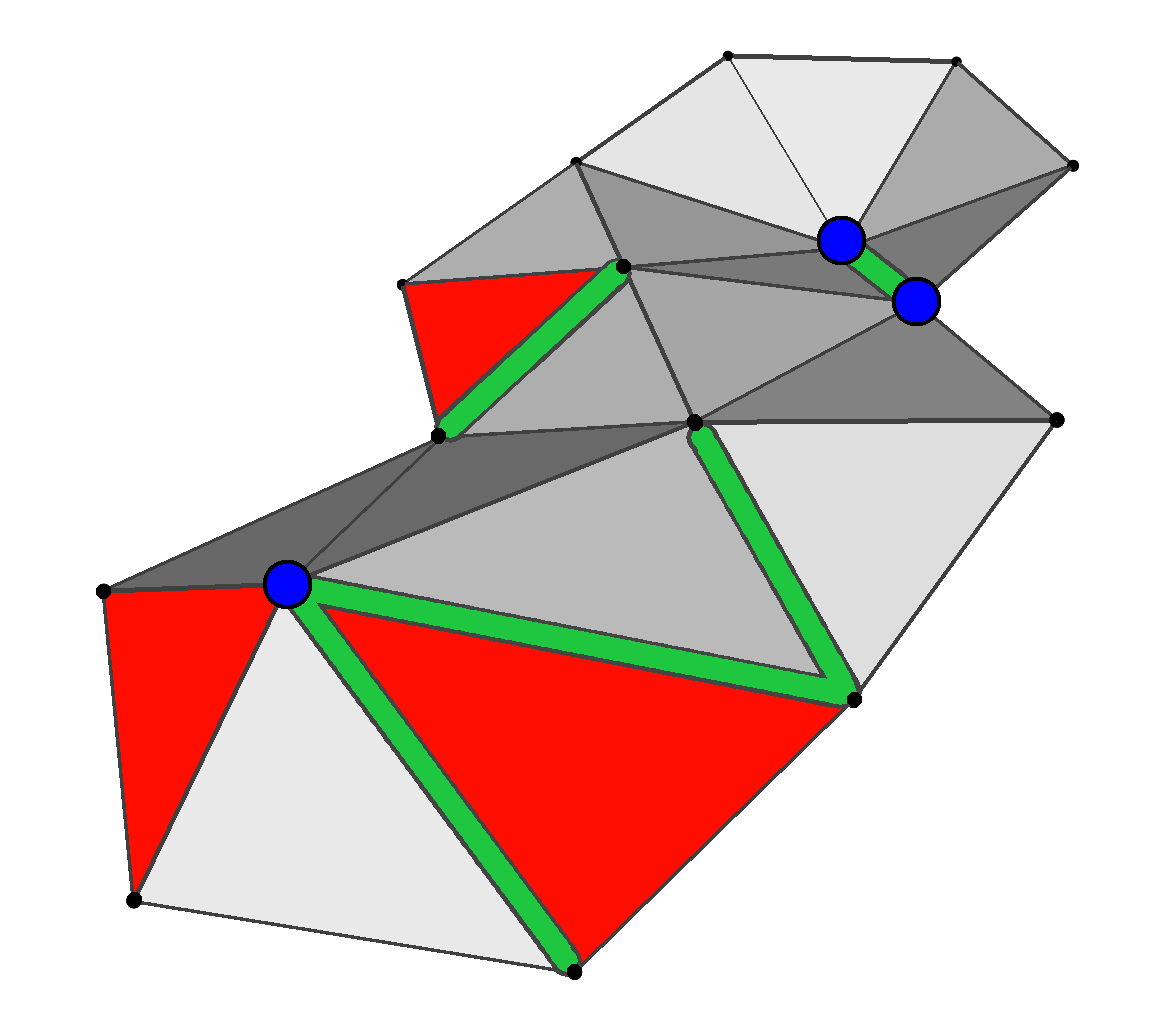} \\
    $(a)\quad\quad\quad$ & $(b)$ & $(c)$ & $(d)$
	\end{tabular}
	\caption{$(a)$ A portion of a 2-dimensional simplicial complex $\Sigma$ endowed with an injective scalar function $f:\Sigma_0 \rightarrow \mathbb{R}$ defined on its vertices $\Sigma_0$. In $(b)-(d)$, $\Sigma$ is depicted as a surface in $\mathbb{R}^3$ by considering the piecewise-linear interpolation of $f$ as a height function defined on all the points of the underlying space of $\Sigma$. In $(b)$, marked vertices depict the PL critical points of $f$ (maxima in red, saddles in green, minima in blue). In $(c)$ and $(d)$, marked simplices depict the discrete critical simplices of two different gradient vector fields (as before, maxima in red, saddles in green, minima in blue).}
	\label{fig:intro}
\end{figure*}


Because of these different standpoints,  the notions of critical points the two approaches prompt are quite different. In the PL case, the various definitions given in the literature for critical points of a PL function on a discretized manifold domain always point to vertices and reflect our common intuition of features like a minimum, a maximum, or a saddle point. In contrast, in  discrete Morse theory, critical points do not consist of just vertices but, more generally, of simplices of any dimension. This way, at a visual level, discrete Morse theory loses its ties with the smooth theory.  On the other hand,  discrete Morse theory is recently gaining much more visibility than the PL one thanks to its combinatorial nature and to its capability in dealing with arbitrary domains rather than only manifolds.

We think that for these  reasons  it is not obvious how to relate and interpret the critical sets obtained by the PL and the combinatorial approaches.
An example of that is depicted in Figure \ref{fig:intro}. Let us consider the 2-dimensional simplicial complex $\Sigma$ endowed with an injective scalar function $f$ defined on its vertices represented in Figure \ref{fig:intro}$(a)$. As shown in Figure \ref{fig:intro}$(b)$, $\Sigma$ can be depicted as a surface in $\mathbb{R}^3$ by considering the piecewise-linear interpolation of $f$ as a height function defined on all the points of the underlying space of $\Sigma$. In accordance with our intuition, the PL critical points of $f$ are localized in the peaks, the saddles and the troughs of the considered region of $\Sigma$. The same is not true in general for the discrete critical simplices of a gradient vector field $V$. For instance, while the discrete critical simplices depicted in Figure \ref{fig:intro}$(c)$ are in 1-to-1 correspondence with the PL critical points of $f$ (moreover, maximum/saddle/minimum PL critical points are in correspondence with maximum/saddle/minimum discrete critical simplices, respectively) and closely located to them (i.e., each PL critical point $v$ is necessarily a vertex of the corresponding discrete critical simplex $\sigma$), such a correspondence is not realized by the discrete critical simplices depicted in Figure \ref{fig:intro}$(d)$.
The described issue represents a real obstruction for experts to exploit the full potentialities of combining the two theories, and for practitioners  to knowingly adopt  in their application domains one or the other of the two theories. The aim of this paper is to unveil such a relation.

\paragraph{Contributions}
Firstly, in Section \ref{sec:critical-definition}, we consider the different definitions of PL critical points given in the literature, starting from the original definition given in \cite{Banchoff67} that has been later specialized or generalized according to specific tasks and working dimensions.
In spite of the intuitive analogy between them, to the best of our knowledge, a formal proof that the proposed definitions are equivalent has never been given, and we  make up for this gap in Section \ref{sec:critical-equivalence}.

Secondly, in Section \ref{sec:critical-correspondence}, we  turn our attention to discrete Morse theory and introduce the notion of  relative perfectness for a discrete gradient vector field $V$ with respect to a given function $f$. It amounts to require that the number of discrete critical simplices of $V$  coincides with the number of topological changes  occurring along  a sublevel set  filtration by $f$. We show the utility of this notion  in Section \ref{subsec:critical-correspondence} by proving that, for relatively perfect discrete gradient vector fields there is  a well-defined  correspondence between its discrete critical simplices and the PL critical points of $f$. In particular, such a correspondence is a bijection for PL Morse functions. Moreover, it ensures that each discrete critical simplex $\sigma$ and its corresponding PL critical point $v$ are closely located:  each PL critical point $v$ is necessarily a vertex of the corresponding discrete critical simplex $\sigma$.

  As a third and final contribution, in Section \ref{subsec:existence} we prove that, for a combinatorial manifold $\Sigma$ of dimension lower than or equal to 3 endowed with a function $f$ defined on the vertices of $\Sigma$, it is always possible to build a discrete gradient vector field $V$ relatively perfect with respect to $f$. The proof exhibits an algorithmic strategy to build such a gradient and, combined with our second contribution, ensures that a correspondence between the critical sets obtained by the PL and the discrete approach is always established for low dimensions.

\paragraph{Related works}
Although the interest in an explicit understanding of the connections between the PL and discrete Morse theories would be most natural, very few papers in the literature address this topic.
The work by Lewiner \cite{Lewiner13} is, to the best of our knowledge, the only work that  deals with it. In his work, the author  adopts a greedy algorithm for the construction of a discrete gradient vector field \cite{Lewiner03} to build, after a sequence of barycentric subdivisions, an adjacent discrete critical simplex for each PL critical point. Even if this represents a first encouraging result, the obtained correspondence is affected by some serious constraints. Firstly, it is limited to simplicial complexes up to dimension 2. Secondly, the entire approach is available only for a specific algorithm. Finally,  the need of  a sequence of barycentric subdivisions to obtain the desired correspondence, with the consequent rapid increase in the number of simplices, is not desirable.

Related to the problems here addressed is also the work by Benedetti \cite{Benedetti12}. Benedetti proves that, taking  a Morse vector to list the number of critical points in each dimension, if a smooth manifold $M$ admits a Morse vector $\mathbf c$, then for any PL triangulation $\Sigma$ of $M$ there exists a finite number of barycentric subdivisions of $\Sigma$ such that the obtained triangulation admits $\mathbf c$ as a discrete Morse vector.

From different perspectives, both Benedetti \cite{Benedetti12} and Lewiner et al. \cite{Lewiner03}  are also interested in perfect functions for which the number of critical points is the minimal one allowed by the homology of the manifold. In applications,  functions correspond to measurements and cannot be chosen, so perfectness is scarcely interesting. Instead, it may be useful   to achieve  the minimal number of critical points ensuring the same persistent homology as the given function. In particular, Robins et al. in \cite{Robins11} show that such optimality is achievable for cubical complexes up to dimension 3. In the present paper, we rephrase this kind of optimality in terms of relative homology, hence calling it {\em relative perfecteness},  and prove that it is also achievable for simplicial complexes up to dimension 3, improving \cite{Landi2019}.

\section{Basic notions in PL and discrete Morse theory}\label{sec:critical-definition}
In this section, we briefly introduce the required background notions on PL and discrete Morse theory with a special emphasis on the definitions of critical points and simplices provided in the literature.

\paragraph{Notations and working hypothesis}
From now on, we will adopt the following notations that will also describe the common framework in which the different versions of the PL Morse theory are settled.
Given a simplicial complex $\Sigma$, we denote by $|\Sigma|$ the underlying space of $\Sigma$, also known as the polytope of $\Sigma$ (i.e., the geometric realization of $\Sigma$ as a subspace of the Euclidean space $\mathbb{R}^n$ where it is embedded).
Hereafter, depending on the domain on which it is applied, the symbol $H_*$ represents  singular or simplicial homology with coefficients in a field. By $\beta_i$,  we denote the rank of the $i^{th}$ homology group $H_i$. Analogously, $\tilde \beta_i$ denotes the rank of the $i^{th}$ reduced homology group $\tilde H_i$.

We assume that an injective scalar function $f:\Sigma_0 \rightarrow \mathbb{R}$ is given on the set $\Sigma_0$ of vertices of $\Sigma$.
It can be extended to two functions: a piecewise-linear function $f_{PL}:|\Sigma| \rightarrow \mathbb{R}$ defined by linear interpolation for all the points of $|\Sigma|$, and a function $f_{max}:\Sigma \rightarrow \mathbb{R}$ defined
by mapping each simplex $\sigma \in \Sigma$ to $\max\{ f(v) \,|\, \text{ vertex } v \text{ face of } \sigma\}$.
Thanks to such function $f_{max}$, it is possible to filter $\Sigma$ through a collection of sublevel sets where, for $l\in\mathbb{R}$, the $l$-sublevel set of $\Sigma$ w.r.t. $f_{max}$ is the simplicial complex $\Sigma^l:=\{ \sigma\in \Sigma \,|\, f_{max}(\sigma)\leq l \}$.

%

Given a simplex $\sigma$ of $\Sigma$, the star and the link of $\sigma$ represent combinatorial counterparts of an open neighbourhood and of its boundary, respectively. Formally, the star of a simplex $\sigma$ of $\Sigma$, $\sta(\sigma)$, is defined as the collection of the cofaces of $\sigma$, while the link of a simplex $\sigma$, $\link(\sigma)$, consists of the collection of the simplices of $\Sigma$ that are faces of an element in $\sta(\sigma)$ which do not intersect $\sigma$.

The function $f_{max}$ allows to define the lower star of a simplex $\sigma$ of $\Sigma$, $\sta^{-}(\sigma)$, as the subset of $\sta(\sigma)$ on which the function $f_{max}$ takes values not greater than $f_{max}(\sigma)$. I.e., $\sta^{-}(\sigma):=\sta(\sigma)\cap \Sigma^{f_{max}(\sigma)}$. Similarly, one can define the lower link of $\sigma$, $\link^{-}(\sigma)$, as the intersection $\link(\sigma)\cap \Sigma^{f_{max}(\sigma)}$. The closure under the face relation of some collection $S$ of simplices (such as the star of a vertex) is denoted by $\overline{S}$ and is the smallest simplicial subcomplex of $\Sigma$ containing  $S$.

\begin{figure*}[ht]
  \centering
  \resizebox{1.0\textwidth}{!}{%
    \hspace{2.5cm}	\begin{tikzpicture}
    \draw[fill=black!20] (0,0)--(0,1)--(1,0.5);
    \draw[fill=black!20] (1,2)--(0,1)--(1,0.5);
    \draw[fill=black!20] (1,2)--(2,1)--(1,0.5);
    \draw[fill=black!20] (2,0)--(2,1)--(1,0.5);
    \draw[fill=black!20] (2,0)--(1,-1)--(1,0.5);
    \draw[fill=black!20] (0,0)--(1,-1)--(1,0.5);
    \draw (0,0)--(0,1);
    \draw (0,1)--(1,2);
    \draw (1,2)--(2,1);
    \draw (2,1) -- (2,0);
    \draw (2,0) -- (1,-1);
    \draw (1,-1) -- (0,0);
    \draw (1,0.5) node[circle,fill,inner sep=1.5pt] {};
    \draw (0,0) node[circle,fill,inner sep=1.5pt] {};
    \node[inner sep=0,anchor=west,text width=3.3cm] (note1) at (-0.3,0) {$2$};
    \draw (0,1) node[circle,fill,inner sep=1.5pt] {};
    \node[inner sep=0,anchor=west,text width=3.3cm] (note1) at (-0.3,1) {$1$};
    \draw (1,2) node[circle,fill,inner sep=1.5pt] {};
    \node[inner sep=0,anchor=west,text width=3.3cm] (note1) at (0.7,2) {$7$};
    \draw (2,1) node[circle,fill,inner sep=1.5pt] {};
    \node[inner sep=0,anchor=west,text width=3.3cm] (note1) at (2.3,1) {$3$};
    \draw (2,0) node[circle,fill,inner sep=1.5pt] {};
    \node[inner sep=0,anchor=west,text width=3.3cm] (note1) at (2.3,0) {$8$};
    \draw (1,-1) node[circle,fill,inner sep=1.5pt] {};
    \node[inner sep=0,anchor=west,text width=3.3cm] (note1) at (0.7,-1) {$6$};
    \node[inner sep=0,anchor=west,text width=3.3cm] (note1) at (1.3,0.5) {$5$};
    \draw (1,0.5) -- (0,0);
    \draw (1,0.5) -- (0,1);
    \draw (1,0.5) -- (1,2);
    \draw (1,0.5) -- (2,1);
    \draw (1,0.5) -- (2,0);
    \draw (1,0.5) -- (1,-1);
    \node[inner sep=0,anchor=west,text width=3.3cm] (note1) at (0.8,-1.5) {$(a)$};
    \end{tikzpicture}
    \hspace{-2cm}
    \quad
    \begin{tikzpicture}
    \draw[fill=black!20] (0,0)--(0,1)--(1,0.5);
    \draw[red] (0,0)--(0,1);
    \draw[dotted] (0,1)--(1,2);
    \draw[dotted] (1,2)--(2,1);
    \draw[dotted] (2,1) -- (2,0);
    \draw [dotted](2,0) -- (1,-1);
    \draw[dotted] (1,-1) -- (0,0);
    \draw (1,0.5) node[circle,fill,inner sep=1.5pt] {};

    \node[inner sep=0,anchor=west,text width=3.3cm] (note1) at (-0.3,0) {$2$};
    \draw (0,1) node[circle,fill,red,inner sep=1.5pt] {};
    \node[inner sep=0,anchor=west,text width=3.3cm] (note1) at (-0.3,1) {$1$};

    \node[inner sep=0,anchor=west,text width=3.3cm] (note1) at (0.7,2) {$7$};

    \node[inner sep=0,anchor=west,text width=3.3cm] (note1) at (2.3,1) {$3$};

    \node[inner sep=0,anchor=west,text width=3.3cm] (note1) at (2.3,0) {$8$};

    \node[inner sep=0,anchor=west,text width=3.3cm] (note1) at (0.7,-1) {$6$};
    \node[inner sep=0,anchor=west,text width=3.3cm] (note1) at (1.3,0.5) {$5$};
    \draw (1,0.5) -- (0,0);
    \draw (1,0.5) -- (0,1);
    \draw [dotted](1,0.5) -- (1,2);
    \draw (1,0.5) -- (2,1);
    \draw[dotted] (1,0.5) -- (2,0);
    \draw[dotted] (1,0.5) -- (1,-1);
    \node[inner sep=0,anchor=west,text width=3.3cm] (note1) at (0.8,-1.5) {$(b)$};
    \draw (0,0) node[circle,fill,red,inner sep=1.5pt] {};
    \draw (0,1) node[circle,fill,red,inner sep=1.5pt] {};
    \draw (2,1) node[circle,fill,red,inner sep=1.5pt] {};
    \draw (1,2) node[circle,fill, lightgray, inner sep=1.5pt] {};
    \draw (2,0) node[circle,fill, lightgray,inner sep=1.5pt] {};
    \draw (1,-1) node[circle,fill, lightgray,inner sep=1.5pt] {};
    \end{tikzpicture}
    \hspace{-2cm}
    \quad
    \begin{tikzpicture}
    \draw[fill=black!20] (1,2)--(0,1)--(1,0.5);
    \draw[fill=black!20] (1,2)--(2,1)--(1,0.5);
    \draw[fill=black!20] (2,0)--(2,1)--(1,0.5);
    \draw[fill=black!20] (0,0)--(1,-1)--(1,0.5);
    \draw [dotted](0,0)--(0,1);
    \draw (0,1)--(1,2);
    \draw (1,2)--(2,1);
    \draw(2,1) -- (2,0);
    \draw [dotted](2,0) -- (1,-1);
    \draw(1,-1) -- (0,0);
    \draw (1,0.5) node[circle,fill,inner sep=1.5pt] {};
    \draw (0,0) node[circle,fill,inner sep=1.5pt] {};
    \node[inner sep=0,anchor=west,text width=3.3cm] (note1) at (-0.3,0) {$2$};
    \draw (0,1) node[circle,fill,inner sep=1.5pt] {};
    \node[inner sep=0,anchor=west,text width=3.3cm] (note1) at (-0.3,1) {$1$};
    \draw (1,2) node[circle,fill, inner sep=1.5pt] {};
    \node[inner sep=0,anchor=west,text width=3.3cm] (note1) at (0.7,2) {$7$};
    \draw (2,1) node[circle,fill,inner sep=1.5pt] {};
    \node[inner sep=0,anchor=west,text width=3.3cm] (note1) at (2.3,1) {$3$};
    \draw (2,0) node[circle,fill,inner sep=1.5pt] {};
    \node[inner sep=0,anchor=west,text width=3.3cm] (note1) at (2.3,0) {$8$};
    \draw (1,-1) node[circle,fill,inner sep=1.5pt] {};
    \node[inner sep=0,anchor=west,text width=3.3cm] (note1) at (0.7,-1) {$6$};
    \node[inner sep=0,anchor=west,text width=3.3cm] (note1) at (1.3,0.5) {$5$};
    \draw (1,0.5) -- (0,0);
    \draw (1,0.5) -- (0,1);
    \draw (1,0.5) -- (1,2);
    \draw (1,0.5) -- (2,1);
    \draw (1,0.5) -- (2,0);
    \draw (1,0.5) -- (1,-1);
    \node[inner sep=0,anchor=west,text width=3.3cm] (note1) at (0.8,-1.5) {$(c)$};
    \end{tikzpicture}
  } %

  \caption{$(a)$ The star of $v:=f^{-1}(5)$. $(b)$ The lower star of $v$ and its lower link (in red). $(c)$ The triangles in the star of $v$ for which $v$ is middle.}
  \label{fig:saddle}
\end{figure*}
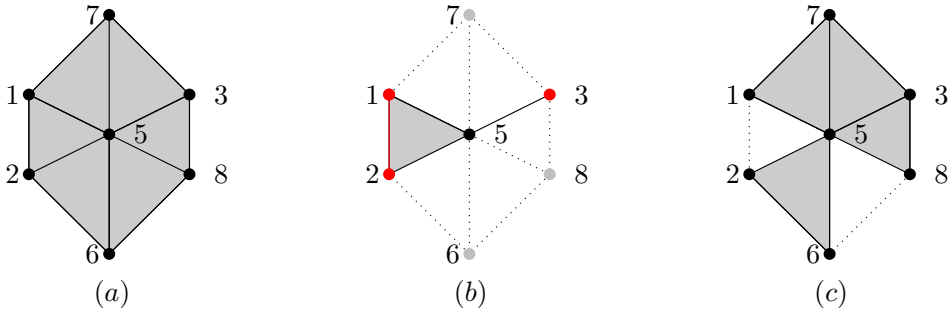

\subsection{PL critical points}\label{subsec:PL}
As previously mentioned, the literature about PL Morse theory proposes several definitions of a critical point. In view of proving their equivalence, we start reviewing these definitions.

Similarly to the smooth theory, PL Morse theory requires working with a  manifold, though  a {\em combinatorial $d$-manifold}, i.e., a simplicial complex of some dimension $d\ge 1$ such that the underlying space of the link of each vertex is homeomorphic to the $(d-1)$-sphere $S^{d-1}$. So, in the following, if  not differently specified, we will always assume that $\Sigma$ is a combinatorial $d$-manifold.

Our brief survey will consider first two definitions of a PL critical point for the case $d=2$ (a widely studied case because of its applications in terrain analysis), and then other two definitions for the case of a combinatorial manifolds of arbitrary dimension $d$.

\paragraph{Banchoff \cite{Banchoff67}}
Banchoff proposes a definition of PL critical points of an injective function $f$ defined on the vertices of a combinatorial $2$-manifold $\Sigma$.
Let $\sigma$ be a triangle $[u, v, w]$ in $\sta(v)$.
If $f(u)<f(v)<f(w)$, we say that $\sigma$ has $v$ {\em middle} for $f$.
For a vertex $v$, we set
$$\iota(v, f):=1-\frac{1}{2}\cdot\#\left\{\text{triangles in $\sta(v)$ with $v$ middle for $f$}\right\}.$$
Hence, vertices of $\Sigma$ are classified as follows:
$$\iota(v, f)=
\begin{cases}
1 & \leftrightarrow \text{ $v$ is a {\em minimum} or {\em maximum},}\\
0 & \leftrightarrow \text{ $v$ is a {\em regular} point,}\\
-k<0 & \leftrightarrow \text{ $v$ is a {\em saddle} of {\em multiplicity $k$}.}
\end{cases}
$$

\paragraph{Edelsbrunner et al. \cite{Edelsbrunner01}}
Edelsbrunner et al. introduce a different definition of a PL critical point of an injective function $f$ defined on the vertices of a combinatorial $2$-manifold $\Sigma$. A vertex $v$ of $\Sigma$ is declared critical or not depending on the number of ``wedges" in which the lower star of $v$ is subdivided.
Formally,  a {\em section} of $\sta^{-}(v)$ is an edge or a triangle in $\sta^{-}(v)$.  Let $S$ be a collection of sections in $\sta^{-}(v)$.   $S$ is called  a {\em contiguous section} of $\sta^{-}(v)$ if $S\setminus\{v\}$ is connected.
A {\em wedge} of $\sta^{-}(v)$ is defined as a contiguous section of $\sta^{-}(v)$ whose boundary in the lower link of $v$ is not a cycle.
Letting $W$ the number of wedges of $\sta^{-}(v)$, $v$ is classified as follows:
$$W=
\begin{cases}
0 & \leftrightarrow \text{ $v$ is a {\em minimum} or {\em maximum},}\\
1 & \leftrightarrow \text{ $v$ is a {\em regular} point,}\\
k+1>1 & \leftrightarrow \text{ $v$ is a {\em saddle} of {\em multiplicity $k$}.}
\end{cases}
$$
For $W=0$, we can distinguish a minimum or maximum point according to the fact that $\sta^{-}(v)$ is $\{v\}$ or it coincides with the entire $\sta(v)$.

\paragraph{Brehm and K{\"u}hnel \cite{Brehm87}}
A definition of a PL critical point for the case of a combinatorial manifold $\Sigma$ of arbitrary dimension $d$ has been proposed in \cite{Brehm87}.
Intuitively, the authors define a vertex $v$ of $\Sigma$ as critical by checking if there is a change in homology when one removes $v$ from its sublevel set.
Formally, a vertex $v$ of $\Sigma$ is classified as {\em PL critical} for $f$ whenever the relative homology $H_*(|\Sigma^l|, |\Sigma^l|\setminus\{v\})$,
with $l=f(v)$, is non-trivial. Otherwise, $v$ is called {\em regular}.
Thanks to the following isomorphisms, the criterion can be expressed in various equivalent ways:
\begin{eqnarray}\label{brehm}
\begin{aligned}
H_*(|\Sigma^l|, |\Sigma^l|\setminus\{v\})&\cong& H_*(|\Sigma^l|, |\Sigma^l\setminus \sta(v)|)\\
&\cong& H_*(\Sigma^l, \Sigma^l\setminus \sta(v))\\
&\cong& H_*(\Sigma^l\cap \overline{\sta(v)}, \Sigma^l\cap \link(v))\\
&\cong& H_*(\overline{\sta^{-}(v)}, \link^{-}(v)).
\end{aligned}
\end{eqnarray}

A PL critical point $v$ of $f$ is said to have {\em index $i$} and {\em multiplicity $k_i$} if $\beta_i(|\Sigma^l|, |\Sigma^l|\setminus\{v\})=k_i$.
In general, a PL critical point might be critical with respect to several indices and its {\em total multiplicity} is $k:=\sum_{i=0}^d k_i$.
PL critical point of index 0 or $d$ will be called point of {\em minimum} or {\em maximum}, respectively. Other PL critical points will be addressed as {\em saddle} points.

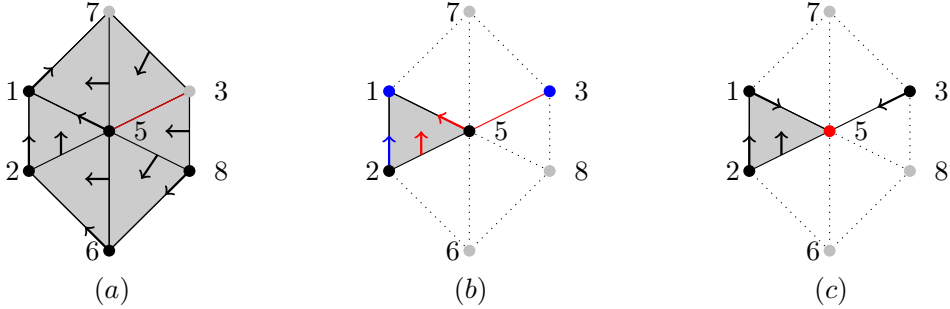
\begin{figure*}[ht]
  \centering
  \resizebox{1.0\textwidth}{!}{%
    \hspace{2.5cm}
    \begin{tikzpicture}
    \draw[fill=black!20] (0,0)--(0,1)--(1,0.5);
    \draw[fill=black!20] (1,2)--(0,1)--(1,0.5);
    \draw[fill=black!20] (1,2)--(2,1)--(1,0.5);
    \draw[fill=black!20] (2,0)--(2,1)--(1,0.5);
    \draw[fill=black!20] (2,0)--(1,-1)--(1,0.5);
    \draw[fill=black!20] (0,0)--(1,-1)--(1,0.5);
    \draw (0,0)--(0,1);
    \draw (0,1)--(1,2);
    \draw (1,2)--(2,1);
    \draw (2,1) -- (2,0);
    \draw (2,0) -- (1,-1);
    \draw (1,-1) -- (0,0);
    \draw [red] (1,0.5) -- (2,1);
    \draw (1,0.5) node[circle,fill,inner sep=1.5pt] {};
    \draw (0,0) node[circle,fill,inner sep=1.5pt] {};
    \node[inner sep=0,anchor=west,text width=3.3cm] (note1) at (-0.3,0) {$2$};
    \draw (0,1) node[circle,fill,inner sep=1.5pt] {};
    \node[inner sep=0,anchor=west,text width=3.3cm] (note1) at (-0.3,1) {$1$};

    \node[inner sep=0,anchor=west,text width=3.3cm] (note1) at (0.7,2) {$7$};
    \draw (2,1) node[circle,fill,lightgray,inner sep=1.5pt] {};
    \node[inner sep=0,anchor=west,text width=3.3cm] (note1) at (2.3,1) {$3$};
    \draw (2,0) node[circle,fill,inner sep=1.5pt] {};
    \node[inner sep=0,anchor=west,text width=3.3cm] (note1) at (2.3,0) {$8$};
    \draw (1,-1) node[circle,fill,inner sep=1.5pt] {};
    \node[inner sep=0,anchor=west,text width=3.3cm] (note1) at (0.7,-1) {$6$};
    \node[inner sep=0,anchor=west,text width=3.3cm] (note1) at (1.3,0.5) {$5$};
    \draw (1,0.5) -- (0,0);
    \draw (1,0.5) -- (0,1);
    \draw (1,0.5) -- (1,2);
    \draw (1,0.5) -- (2,0);
    \draw (1,0.5) -- (1,-1);
    \draw (1,2) node[circle,fill,lightgray,inner sep=1.5pt] {};
    \draw[thick,->] (0,0) -- (0,0.45);
    \draw[thick,->] (0.4,0.2) -- (0.4,0.5);
    \draw[thick,->] (1,0.5) -- (0.6,0.7);
    \draw[thick,->] (1,1.1) -- (0.7,1.1);
    \draw[thick,->] (1.5,1.5) -- (1.35,1.2);
    \draw[thick,->] (2,0.5) -- (1.7,0.5);
    \draw[thick,->] (1.6,0.2) -- (1.4,-0.1);
    \draw[thick,->] (2,0) -- (1.7,-0.3);
    \draw[thick,->] (1,-1) -- (0.7,-0.7);
    \draw[thick,->] (1,-0.1) -- (0.7,-0.1);
    \draw[thick,->] (0,1) -- (0.3,1.3);
    \node[inner sep=0,anchor=west,text width=3.3cm] (note1) at (0.8,-1.5) {$(a)$};
    \end{tikzpicture}
    \hspace{-2cm}
    \quad
      \begin{tikzpicture}
    \draw[fill=black!20] (0,0)--(0,1)--(1,0.5);
    \draw (0,0)--(0,1);
    \draw[dotted] (0,1)--(1,2);
    \draw[dotted] (1,2)--(2,1);
    \draw[dotted] (2,1) -- (2,0);
    \draw [dotted](2,0) -- (1,-1);
    \draw[dotted] (1,-1) -- (0,0);
    \draw [red] (1,0.5) -- (2,1);

    \node[inner sep=0,anchor=west,text width=3.3cm] (note1) at (-0.3,0) {$2$};

    \node[inner sep=0,anchor=west,text width=3.3cm] (note1) at (-0.3,1) {$1$};

    \node[inner sep=0,anchor=west,text width=3.3cm] (note1) at (0.7,2) {$7$};
    \draw (2,1) node[circle,fill,blue,inner sep=1.5pt] {};
    \node[inner sep=0,anchor=west,text width=3.3cm] (note1) at (2.3,1) {$3$};

    \node[inner sep=0,anchor=west,text width=3.3cm] (note1) at (2.3,0) {$8$};

    \node[inner sep=0,anchor=west,text width=3.3cm] (note1) at (0.7,-1) {$6$};
    \node[inner sep=0,anchor=west,text width=3.3cm] (note1) at (1.3,0.5) {$5$};
    \draw (1,0.5) -- (0,0);
    \draw (1,0.5) -- (0,1);
    \draw [dotted](1,0.5) -- (1,2);

    \draw[dotted] (1,0.5) -- (2,0);
    \draw[dotted] (1,0.5) -- (1,-1);
    \draw[thick,->, red] (0.4,0.2) -- (0.4,0.5);
    \draw[thick,->, blue] (0,0) -- (0,0.45);
    \draw[thick,->, red] (1,0.5) -- (0.6,0.7);
    \node[inner sep=0,anchor=west,text width=3.3cm] (note1) at (0.8,-1.5) {$(b)$};
    \draw (1,-1) node[circle,fill, lightgray,inner sep=1.5pt] {};
    \draw (2,0) node[circle,fill, lightgray,inner sep=1.5pt] {};
    \draw (1,2) node[circle,fill, lightgray, inner sep=1.5pt] {};
    \draw (0,1) node[circle,fill,blue,inner sep=1.5pt] {};
    \draw (0,0) node[circle,fill,inner sep=1.5pt] {};
    \draw (1,0.5) node[circle,fill,inner sep=1.5pt] {};
    \end{tikzpicture}

    \hspace{-2cm}
    \quad
    \begin{tikzpicture}
    \draw[fill=black!20] (0,0)--(0,1)--(1,0.5);
    \draw (0,0)--(0,1);
    \draw[dotted] (0,1)--(1,2);
    \draw[dotted] (1,2)--(2,1);
    \draw[dotted] (2,1) -- (2,0);
    \draw [dotted](2,0) -- (1,-1);
    \draw[dotted] (1,-1) -- (0,0);

    \draw (0,0) node[circle,fill,inner sep=1.5pt] {};
    \node[inner sep=0,anchor=west,text width=3.3cm] (note1) at (-0.3,0) {$2$};
    \draw (0,1) node[circle,fill,inner sep=1.5pt] {};
    \node[inner sep=0,anchor=west,text width=3.3cm] (note1) at (-0.3,1) {$1$};

    \node[inner sep=0,anchor=west,text width=3.3cm] (note1) at (0.7,2) {$7$};
    \draw (2,1) node[circle,fill,inner sep=1.5pt] {};
    \node[inner sep=0,anchor=west,text width=3.3cm] (note1) at (2.3,1) {$3$};

    \node[inner sep=0,anchor=west,text width=3.3cm] (note1) at (2.3,0) {$8$};

    \node[inner sep=0,anchor=west,text width=3.3cm] (note1) at (0.7,-1) {$6$};
    \node[inner sep=0,anchor=west,text width=3.3cm] (note1) at (1.3,0.5) {$5$};
    \draw (1,0.5) -- (0,0);
    \draw (1,0.5) -- (0,1);
    \draw [dotted](1,0.5) -- (1,2);
    \draw  (1,0.5) -- (2,1);
    \draw[dotted] (1,0.5) -- (2,0);
    \draw[dotted] (1,0.5) -- (1,-1);
    \draw[thick,->] (0.4,0.2) -- (0.4,0.5);
    \draw[thick,->] (0,0) -- (0,0.45);
    \draw[thick,->] (2,1) -- (1.6,0.8);
    \draw[thick,->] (0,1) -- (0.4,0.8);
    \node[inner sep=0,anchor=west,text width=3.3cm] (note1) at (0.8,-1.5) {$(c)$};
    \draw (1,0.5) node[circle,fill,red, inner sep=1.5pt] {};
    \draw (1,2) node[circle,fill, lightgray, inner sep=1.5pt] {};
    \draw (2,0) node[circle,fill, lightgray,inner sep=1.5pt] {};
    \draw (1,-1) node[circle,fill, lightgray,inner sep=1.5pt] {};
    \end{tikzpicture}
  } %
  \caption{$(a)$ An RP discrete gradient vector field. $(b)$ An RP discrete gradient vector field on the lower star of $v:=f^{-1}(5)$. $(c)$ A non-RP discrete vector field. }
  \label{fig:correspondence}
\end{figure*}

\paragraph{Edelsbrunner et al. \cite{Edelsbrunner10,Edelsbrunner03}}
Another definition of a PL critical point of a function $f$ for a combinatorial $d$-manifold $\Sigma$  has been introduced   for the case $d=3$ in \cite{Edelsbrunner03}, and generalized to  arbitrary dimension in  \cite{Edelsbrunner10}.

Edelsbrunner et al. define a vertex $v$  as PL critical or not depending on the reduced homology of its lower link.
More formally, let $\tilde{\beta}_j$ be the rank of the reduced $j^{th}$ homology group of $\link^{-}(v)$.
A vertex $v$ of $\Sigma$ is called {\em regular} if $\tilde{\beta}_j=0$ for any $j=-1, 0, 1, \dots, d$.
Else, $v$ is called a {\em PL critical point} of {\em index $i$} and {\em multiplicity $k$} of $f$ if
$$\tilde{\beta}_j=
\begin{cases}
k & \text{ for } j=i-1,\\
0 & \text{ otherwise.}
\end{cases}$$
Specifically, a PL critical point of index $i$ is called a {\em minimum} if $i=0$, a {\em maximum} if $i=d$, and an {\em $i$-saddle} otherwise.
A  PL critical point with multiplicity $k>1$ is called a {\em multiple saddle}.
  The function $f$ is called {\em PL Morse} if all its PL critical points have multiplicity 1.\\

Having so many different definitions for a PL critical point, until we will show that they are all equivalent, for the sake of clarity we  distinguish between them by referring to them as either an $I$-critical point ($I$ for index), a $W$-critical point ($W$ for wedge), an  $H$-critical point ($H$ for homology), or an $L$-critical point ($L$ for link), according to whether it is a PL critical point as defined by Banchoff \cite{Banchoff67}, Edelsbrunner et al. \cite{Edelsbrunner01},  Brehm and K{\"u}hnel \cite{Brehm87}, or  Edelsbrunner et al. \cite{Edelsbrunner10}, respectively.

Figure \ref{fig:saddle} illustrates how to classify a vertex of a combinatorial $2$-manifold $\Sigma$ according to the above definitions.
In Figure \ref{fig:saddle}$(a)$, the star of the vertex $v:=f^{-1}(5)$ is given. In Figure \ref{fig:saddle}$(b)$, we see  that $\beta_1(\overline{\sta^{-}(v)}, \link^{-}(v))=1$, that is,  $\beta_1(|\Sigma^5|, |\Sigma^5|\setminus\{v\})=1$.  So, $v$ is an $H$-saddle point. 
Since the lower star has two wedges (one consisting of the triangle $[1, 2, 5]$ and its egdes $[1, 5]$ and $[2, 5]$, and the other one consisting of the edge $[3, 5]$), $v$ is also a $W$-saddle point.  
In Figure \ref{fig:saddle}$(c)$, we have that $\iota(v, f)=1-\frac{1}{2}\cdot4=-1$.  Thus, $v$ is also an $I$-saddle point. 

\subsection{Discrete critical simplices}\label{subsec:discrete}

\textit{Discrete Morse theory} introduced by Forman in \cite{Forman98} represents the most recently proposed discrete counterpart of the smooth Morse theory.
At the price of being a little less intuitive, discrete Morse theory presents some advantages compared to PL Morse theory.
First of all, discrete Morse theory can be defined for arbitrary cell complexes not necessarily discretizing a manifold domain. In spite of this, for the sake of simplicity, we will review discrete Morse theory in the context of simplicial complexes and we will be forced to work in the common framework of the combinatorial manifolds everytime that a direct comparison between the two theories will be presented.
Another great advantage of discrete Morse theory is related to the possibility of describing it in purely combinatorial terms preventing the need of explicitly exhibiting a Morse function defined on the complex. Exploiting this fact, in this subsection we introduce some basic notions of discrete Morse theory by adopting a combinatiorial point of view.

Given a simplicial complex $\Sigma$, discrete Morse theory is based on the definition of a collection of simplex pairs simulating the gradient of a function defined on $\Sigma$.
Formally, a \textit{discrete vector field $V$} on $\Sigma$ is a collection of pairs of simplices $(\sigma, \tau)\in \Sigma \times \Sigma$ such that $\sigma$ is a face of $\tau$ of dimension $\dim(\tau) - 1$ (in the following, we will denote that by $\sigma < \tau$) and each simplex of $\Sigma$ is in at most one pair of $V$. Pictorially, this is illustrated by an arrow from $\sigma$ to $\tau$ as in Figure \ref{fig:correspondence}$(a)$.

Given a discrete vector field $V$, a \textit{$V$-path} (or, equivalently, a \textit{gradient path}) is a sequence $(\sigma_1, \tau_1), (\sigma_2, \tau_2), \dots, (\sigma_{r},\tau_r)$ of pairs of $i$-simplices $\sigma_j$ and $(i+1)$-simplices $\tau_j$, such that $(\sigma_j, \tau_j)\in V$ for $j=1, \dots, r$, and $\sigma_{j+1} < \tau_j$ and $\sigma_j\neq\sigma_{j+1}$ for $j=1, \dots, r-1$.

A $V$-path is a \textit{closed path} if $\sigma_{1}$ is a face of $\tau_{r}$ different from $\sigma_{r}$.
A discrete vector field $V$ is called a {\em discrete gradient vector field} if $V$ is free of closed paths. It is possible (but not necessary) to a define discrete Morse function whose gradient is $V$.

Given a simplicial complex $\Sigma$ endowed with a discrete gradient vector field $V$, an $i$-simplex $\sigma \in \Sigma$ is called {\em regular} if it belongs to a pair of $V$. Otherwise, $\sigma$ is called a \textit{discrete critical simplex of index $i$} (equivalently, a \textit{discrete critical $i$-simplex} or an \textit{$i$-saddle}). More specifically, a discrete critical simplex of index $0$ is called a \textit{minimum}, while a discrete critical simplex of index $d=\dim(\Sigma)$ a \textit{maximum}.

PL and discrete Morse theories deserve to be addressed as discretized versions of Morse theory since they both adapt the fundamental theorems and properties holding in the smooth case to the  combinatorial setting. Among these results, there is a collection of inequalities usually called {\em weak Morse inequalities}.
Precisely, in discrete Morse theory, given a discrete gradient vector field $V$ on a simplicial complex $\Sigma$, weak Morse inequalities state that $m_i(V)\geq \beta_i(\Sigma)$, for any $i=0, \dots, \dim(\Sigma)$,
where $m_i(V)$ denotes the number of discrete critical $i$-simplices of $V$. Analogous inequalities hold for a PL Morse function $f$.
A discrete gradient vector field $V$ is called {\em perfect} if, for any $i=0, \dots, \dim(\Sigma)$, the equality $m_i(V)= \beta_i(\Sigma)$ is satisfied.

\section{Equivalence between the notions of PL critical points}\label{sec:critical-equivalence}
This section is mainly devoted to proving that all the notions of a  PL critical point presented in the previous section are equivalent.
More formally, the main contribution of this section is the following result.
\begin{theorem}\label{thm:eq}
For a vertex $v$ of a combinatorial manifold $\Sigma$ of arbitrary dimension $d\ge 1$,  endowed with an injective scalar function $f$ defined on its vertices, the following statements are equivalent:
\begin{enumerate}
\item $v$ is an $H$-critical point of $f$ of index $i$ and multiplicity $k_i$ 
\item  $v$ is an $L$-critical point of $f$ of the same index and the same multiplicity.           
\end{enumerate}
In particular, for $d=2$, being an $H$-critical point of $f$ is also equivalent to being an $I$-critical point or  a $W$-critical point.
\end{theorem}


As a side result, we also show that, even if the definition of an  $H$-critical point  can be easily generalized to any point of the underlying space of $\Sigma$, it is not actually possible for a point $p$ of $|\Sigma|$ that is not a vertex of $\Sigma$ to be so.

\subsection{Equivalence}\label{subsec:equivalence}

A required first step for establishing a correspondence between critical sets is the proof of the equivalence of the various notions of a PL critical point claimed by Theorem \ref{thm:eq}.



Let us preliminarily notice that in the papers by Banchoff \cite{Banchoff67} and by Edelsbrunner et al. \cite{Edelsbrunner01,Edelsbrunner10,Edelsbrunner03} the only PL critical points which are addressed as multiple are the saddle points. In contrast, in the classification proposed Brehm and K{\"u}hnel \cite{Brehm87}, the existence of non-saddle PL critical points with total multiplicity greater than 1 is not explicitly excluded.
The following results confirm the non-existence of such points and ensure us that all the proposed classifications are complete.

\begin{lemma}\label{equation}
  Given a vertex $v$ of $\Sigma$, for every $i$, we have that:
  \begin{eqnarray*}
    H_i(\overline{\sta^{-}(v)}, \link^{-}(v))\cong \tilde{H}_{i-1}(\link^{-}(v)).
  \end{eqnarray*}
\end{lemma}

\begin{proof}
  Thanks to the properties of relative homology, we have the following long exact sequence of the pair $(\overline{\sta^{-}(v)}, \link^{-}(v))$ for reduced homology
  {\small
  \begin{gather*}
    \cdots \rightarrow \tilde{H}_{i}(\link^{-}(v)) \rightarrow \tilde{H}_{i}(\overline{\sta^{-}(v)})  \rightarrow \tilde{H}_i(\overline{\sta^{-}(v)}, \link^{-}(v)) \rightarrow \\
    \rightarrow \tilde{H}_{i-1}(\link^{-}(v)) \rightarrow \tilde{H}_{i-1}(\overline{\sta^{-}(v)}) \rightarrow \cdots
  \end{gather*}}
  as well as the isomorphism $$H_i(\overline{\sta^{-}(v)}, \link^{-}(v))\cong \tilde{H}_i(\overline{\sta^{-}(v)}, \link^{-}(v)).$$
  Moreover, since $\overline{\sta^{-}(v)}$ is a cone, $\tilde{H}_i(\overline{\sta^{-}(v)})=0$ for every $i$.
  By combining the above facts, the thesis follows.
\qed \end{proof}

\begin{lemma}\label{multiplicity}
  Let $v$ be an $H$-critical point of index $i$ with total multiplicity $k:=\sum_{i=0}^d k_i>1$. 
 Then, $k_0=k_d=0$.
\end{lemma}

\begin{proof}
  Let us prove that by reductio ad absurdum. Let us suppose that $k_0\neq 0$. This is true if and only if (by Lemma \ref{equation}) $\tilde{\beta}_{-1}(\link^{-}(v))\neq 0$ if and only if $\link^{-}(v)=\emptyset$ if and only if $\sta^{-}(v)=\{v\}$. Thus, denoting $f(v)$ by $l$,
  $${\small \beta_i(|\Sigma^l|, |\Sigma^l|\setminus\{v\})=\beta_i(\overline{\sta^{-}(v)}, \link^{-}(v))=\begin{cases}
  1 & \text{ for } i=0,\\
  0 & \text{ otherwise.}
  \end{cases}}$$
  This leads to a contradiction since by hypothesis $k>1$.
  Now, let us suppose that $k_d\neq 0$. This is true if and only if (by Lemma \ref{equation}) $\tilde{\beta}_{d-1}(\link^{-}(v))\neq 0$ if and only if (by the fact that $\Sigma$ is a combinatorial $d$-manifold and so the underlying space of the link of each vertex is homeomorphic to the $(d-1)$-sphere $S^{d-1}$) $\link^{-}(v)=\link(v)=S^{d-1}$ if and only if $\sta^{-}(v)=\sta(v)$. Thus,
  $${\small \beta_i(|\Sigma^l|, |\Sigma^l|\setminus\{v\})=\beta_i(\overline{\sta^{-}(v)}, \link^{-}(v))=\begin{cases}
  1 & \text{ for } i=d,\\
  0 & \text{ otherwise.}
  \end{cases}}$$
  This leads to a contradiction since by hypothesis $k>1$.
\qed \end{proof}


Focusing on the case of a combinatorial $2$-manifold $\Sigma$, we prove here the equivalence between the notions of $I$-, $W$-, and $H$-critical points.

\begin{proposition}\label{eq2d}
  For $d=2$, a vertex $v$ of $\Sigma$ is an {$H$-}critical point of $f$ of index $i$ and multiplicity $k_i$ if and only if $v$ is an $I$- or a $W$-critical point of $f$ of the same index and the same multiplicity.
\end{proposition}

\begin{proof}
  Let us start by comparing the definitions of an $I$- and a $W$-critical point. 
	Their equivalence is easily proved by noticing that by definition, for a vertex $v$, the number  $W$ of wedges of $\sta^-(v)$ coincides with half the number of triangles in $\sta(v)$ with $v$ middle for $f$.
  In order to conclude the proof, we show the equivalence between the definition of a $W$- 
   and an $H$-critical point. 
  Let us consider the various possible cases.
  Let $v$ be a point of minimum according to $W$-criticality. 
   By definition, we have that $W=0$ and $\sta^-(v)=\{v\}$. This is equivalent to the fact that $\link^-(v)$ is empty. Since the only simplicial complex having the same reduced homology of the empty complex is the empty complex itself, thanks to Lemma \ref{equation}, the previous condition is satisfied if and only if $v$ is a PL critical point of index 0 and multiplicity 1 according to $H$-criticality, 
    i.e., $v$ is a point of minimum according to $H$-criticality. 
    A vertex $v$ is a point of maximum according to $W$-criticality 
    if and only if $W=0$ and $\sta^-(v)=\sta(v)$, i.e., if and only if $\link^-(v)=\link(v)$. Since $\Sigma$ is a combinatorial $2$-manifold, this is true if and only if $\link^-(v)$ has the same reduced homology of the sphere $S^{1}$. By Lemma \ref{equation}, this is satisfied if and only if $v$ is an $H$-critical point of index 2 and multiplicity 1, 
    i.e., $v$ is a point of maximum according to $H$-criticality. 
   Given a vertex $v$ which is not a point of minimum or maximum according to $W$-criticality, 
   we have that $W>0$.  Let us notice that, in such case, $\link^-(v)$  is a simplicial complex consisting of exactly $W$ connected components and free of higher dimensional homological cycles: for $W>0$,
  $$\tilde\beta_i(\link^-(v))=\begin{cases} W-1 & \text{ for } i=0,\\
  0 & \text{ otherwise.}
  \end{cases}$$
  So, the number of wedges of $\sta^{-}(v)$ is $W=1$ if and only if $v$ is a regular point according to $H$-criticality; 
    $W>1$ if and only if $v$ is an $H$-critical point of index 1 and multiplicity $W-1$. 
\qed \end{proof}


We are now ready for proving Theorem \ref{thm:eq}.
\begin{proof}
The case $d=2$ has already been discussed and proved by Proposition \ref{eq2d}. For the case of arbitrary $d\geq 1$, the fact that a vertex $v$ of $\Sigma$ is an $H$-critical point of $f$ of index $i$ and multiplicity $k_i$ if and only if $v$ is an $L$-critical point of $f$ of the same index and the same multiplicity is obtained by combining Lemma \ref{equation} and Lemma \ref{multiplicity}.
\qed \end{proof}

Theorem \ref{thm:eq} ensures us an equivalence between all the definitions of PL critical point proposed in the literature.
So, in what follows, it will not be ambiguous to address a vertex as a PL critical point without specifying which definition we are adopting.

\subsection{PL critical points are necessarily vertices}\label{subsec:pointsnotcritical}



The generalization of the definition of critical point given by Brehm and K\"{u}hnel in \cite{Brehm87} to the case when $p$ is not necessarily a vertex of $\Sigma$ could straightforwardly be as follows: $p$ is  PL critical for $f$ if and only if $H_*(|\Sigma|^l, |\Sigma|^l\setminus\{p\})\neq 0$, where $l=f_{PL}(p)$ and, for a subset $X\subseteq |\Sigma|$, $X^l =\{p\in X \, | \, f_{PL}(p)\leq l \}$.
In the particular case when $p$ is a vertex of $\Sigma$, this is equivalent to the Brehm and K\"{u}hnel's definition because $|\Sigma|^l$ deformation retracts onto $|\Sigma^l|$.
On the other hand, we prove here that this generalization would not add any new PL critical point to the function $f$ because only the vertices of $\Sigma$ can be PL critical.

\begin{proposition}
	For every point $p\in|\Sigma|$ such that $p \notin \Sigma_0$,
	$H_i(|\Sigma|^l, |\Sigma|^l\setminus\{p\})=0$, for each $i\geq 0$.
\end{proposition}

\begin{proof}
Let $\sigma$ be the simplex of $\Sigma$ of minimal dimension containing $p$, and let us define $V=\overline{\sta(\sigma)}$.
Since $|\Sigma|^l \setminus |V|^l \subset |\Sigma|^l\setminus\{p\} \subset |\Sigma|^l$ such that the closure of $|\Sigma|^l \setminus |V|^l$ is contained in the interior of $|\Sigma|^l\setminus\{p\}$, by excision theorem, we have, for each $i\geq 0$,  $$H_i(|\Sigma|^l, |\Sigma|^l\setminus\{p\})\cong H_i(|V|^l, |V|^l\setminus\{p\}).$$
Let $\Gamma$ be the stellar subdivision of $V$ at face $\sigma$ such that $p$ is a vertex of $\Gamma$, i.e.,
$$\Gamma=p\ast\partial \sigma\ast\link(\sigma),$$
where $\partial \sigma$ is the set of all proper faces of $\sigma$ and $\ast$ is the usual join operation. Now, for each $i\geq 0$,
\begin{eqnarray*}
\begin{aligned}
	H_i(|V|^l, |V|^l\setminus\{p\})&\cong H_i(|\Gamma|^l, |\Gamma|^l\setminus\{p\})\\
	&\cong H_i(|\Gamma^l|, |\Gamma^l|\setminus\{p\})\\
	&\cong H_i(\overline{\sta_{\Gamma}^{-}(p)}, \link_{\Gamma}^{-}(p))\\
	&\cong \tilde{H}_{i-1}( \link_{\Gamma}^{-}(p)),
\end{aligned}
\end{eqnarray*}
where the first isomorphisms holds because $|V|=|\Gamma|$, the second one because $|\Gamma|^l$ deformation retracts onto $|\Gamma^l|$, the third one by Equation (\ref{brehm}), and the last one by Lemma \ref{equation}.

Now, since $\sta_{\Gamma}(p)= p\ast\partial \sigma\ast\link(\sigma)$, then $\link_{\Gamma}(p)= \partial \sigma\ast\link(\sigma)$.
Let $\gamma$ be the face of maximal dimension of $\sigma$ spanned by the vertices in $\link_{\Gamma}(p)$ with value less than or equal to $l$. Then, $\link_{\Gamma}^{-}(p)=\gamma\ast \link^{-}(\sigma)$. Because $p$ is not a vertex of $\Sigma$, the simplex $\gamma$ is non-empty.  Let $v$ be a vertex of $\gamma$ and $\beta$ be the maximal dimensional face of $\gamma$ which does not contain $v$ and might be empty.  Then, $\gamma= v\ast \beta$.  Since the join operation is associative, $\gamma\ast \link^{-}(\sigma) = v\ast \beta \ast \link^{-}(\sigma) $.  So, $\gamma\ast \link^{-}(\sigma)$ is a cone and therefore it is contractible.
In conclusion, $\tilde{H}_{i-1}( \link_{\Gamma}^{-}(p))\cong 0$, yielding the claim.
\qed \end{proof}

\section{Relating PL and discrete critical sets}\label{sec:critical-correspondence}

The aim of this section is to investigate under which conditions it is possible to establish a correspondence between the set of the PL critical points of a function $f$ on a domain $\Sigma$ and the set of the discrete critical simplices of a discrete gradient vector field on $\Sigma$.
For the rest of the section, unless differently specified, we assume that $\Sigma$ is a combinatorial $d$-manifold,  $f:\Sigma_0 \rightarrow \mathbb{R}$ is an injective function defined on the vertices of $\Sigma$, and  $V$ is a discrete gradient vector field on $\Sigma$.

\subsection{An explicit correspondence}\label{subsec:critical-correspondence}
%
Given a value $l$ in the image of $f$, let us denote by $l'$ the greatest value in the image of $f$ among the ones strictly lower than $l$, if any, and $l':=l-1$, otherwise.  The number $\beta_i(\Sigma^l, \Sigma^{l'})$ denotes the number of variations in the $i^{th}$ homology space occurred at value $l$.\\ 
Equivalently, a simple calculation shows that, denoting by $\phi^l_i$ the homology map in dimension $i$ induced by the inclusion of $\Sigma^{l'}$ into $\Sigma^{l}$, $\beta_i(\Sigma^l, \Sigma^{l'}) = \dim(\ker \, \phi^l_{i-1}) + \dim(\coker \, \phi^l_i)$.

\begin{definition}
  A discrete gradient vector field $V$ on $\Sigma$ is called {\em relatively perfect} (briefly, {\em RP}) w.r.t. $f:\Sigma_0 \rightarrow \mathbb{R}$ if  $m^l_i(V) = \beta_i(\Sigma^l, \Sigma^{l'})$, for every $i\in \mathbb{N}$ and every value $l \in \Ima \, f$,
  where $m^l_i(V)$ denotes the number of discrete critical $i$-simplices for $V$ in $\Sigma^l\setminus \Sigma^{l'}$. 
\end{definition}

Whenever there is no ambiguity about the considered discrete gradient vector field, we will write $m^l_i$ in place of $m^l_i(V)$.

The above definition and the equivalence shown in the previous subsection (Theorem \ref{thm:eq}) enable us to find a correspondence between PL critical points and discrete critical simplices.

\begin{theorem}\label{proposition-corrispondenza}
  Let  $V$ be a relatively perfect discrete gradient vector field on $\Sigma$  w.r.t. $f$. Then, a vertex $v\in\Sigma_0$ is a PL critical point of index $i$ and multiplicity $k_i$ of $f$  if and only if there are exactly $k_i$ discrete critical $i$-simplices $\sigma$ of $V$ such that $\sigma\in\sta(v)$ and $f_{max}(\sigma)=f(v)$.
\end{theorem}

\begin{proof}
  By  definition, a vertex $v$ is  PL critical  of index $i$ and multiplicity $k_i$   if and only if, setting
  $l=f(v)$,
  \begin{eqnarray*}
    \dim(H_j(|\Sigma^l|, |\Sigma^l|\setminus\{v\}))=
    \begin{cases}
      k_i & \text{ for } j=i,\\
      0 & \text{ otherwise.}
    \end{cases}
  \end{eqnarray*}
  Note that   $\Sigma^{l'}=\Sigma^l\setminus \sta(v)$.  Indeed,   $\sigma\in\sta(v)$ implies $f_{max}(\sigma)\ge f(v)=l$. So, on one hand, $\Sigma^{l'}\subseteq \Sigma^l\setminus \sta(v)$ because $l'<l$. On the other hand,  for each $\sigma\in \Sigma^l\setminus \sta(v)$, $f_{\max}(\sigma)<l$ by the injectivity of $f$, yielding $f_{\max}(\sigma)\le l'$ by definition of $l'$. Thus, $\dim(H_j(|\Sigma^l|, |\Sigma^l|\setminus\{v\}))= \beta_j(\Sigma^l, \Sigma^{l'})=m^l_j$, for every $j\in \mathbb{Z}$,
  where the first equality follows from Equation (\ref{brehm})  and the second equality from relatively perfectness of $V$ w.r.t. $f$.
  Thus, the claim follows by recalling that $m^l_j$ is the number of discrete critical $j$-simplices $\sigma$ of $V$ such that $\sigma\in \Sigma^l\setminus \Sigma^{l'}$, $\Sigma^l\setminus \Sigma^{l'}=\sta(v)$, and   $l=f(v)$.
\qed \end{proof}

As a consequence of the above results, we can give the following corollary.

\begin{corollary}\label{cor:arbitraryDIM}
  Let $V$ be a discrete gradient vector field on $\Sigma$ relatively perfect w.r.t. $f$. Then, there is a 1-to-$k_i$ correspondence between  PL critical points of index $i$ and multiplicity $k_i$ of $f$  and discrete critical $i$-simplices $\sigma$ of $V$ such that  $f_{max}(\sigma)=f(v)$. In particular, if $f$ is PL Morse, then the correspondence is bijective.
\end{corollary}

	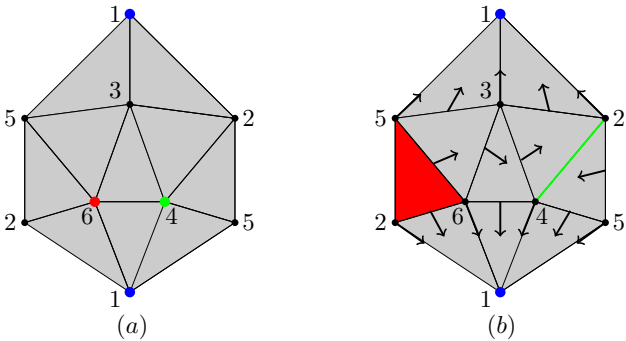
\begin{figure}[hbtp]
		\centering
		\resizebox{0.6\textwidth}{!}{%
			\begin{tikzpicture}
			\draw[fill=black!20] (0,0)--(0,1.5)--(1,0.3);
			\draw[fill=black!20] (1,0.3)--(2,0.3)--(1.5,1.7);
			\draw[fill=black!20] (2,0.3)--(3,0)--(3,1.5);
			\draw[fill=black!20] (2,0.3)--(3,1.5)--(1.5,1.7);
			\draw[fill=black!20] (1.5,1.7)--(3,1.5)--(1.5,3);
			\draw[fill=black!20] (0,1.5)--(1.5,1.7)--(1.5,3);
			\draw[fill=black!20] (0,1.5)--(1,0.3)--(1.5,1.7);
			\draw[fill=black!20] (0,0)--(1,0.3)--(1.5,-1);
			\draw[fill=black!20] (2,0.3)--(1,0.3)--(1.5,-1);
			\draw[fill=black!20] (2,0.3)--(3,0)--(1.5,-1);
			\draw (0,0)--(0,1.5);
			\draw (0,1.5)--(1.5,3);
			\draw (1.5,3)--(3,1.5);
			\draw (3,1.5) -- (3,0);
			\draw (3,0) -- (1.5,-1);
			\draw (1.5,-1) -- (0,0);
			\draw (1.5,3) -- (1.5,1.7);
			\draw (1.5,1.7) -- (0,1.5);
			\draw (1.5,1.7) -- (3,1.5);
			\draw (1,0.3) -- (0,0);
			\draw (1,0.3) -- (0,1.5);
			\draw (1,0.3) -- (1.5,1.7);
			\draw (1,0.3) -- (2,0.3);
			\draw (2,0.3) -- (1.5,1.7);
			\draw (2,0.3) -- (3,1.5);
			\draw (2,0.3) -- (3,0);
			\draw (2,0.3) -- (1.5,-1);
			\draw (1,0.3) -- (1.5,-1);

			\draw (0,1.5) node[circle,fill,inner sep=1pt] {};
			\node[inner sep=0,anchor=west,text width=3.3cm] (note1) at (-0.3,1.5) {$5$};
			\draw (0,0) node[circle,fill,inner sep=1pt] {};
			\node[inner sep=0,anchor=west,text width=3.3cm] (note1) at (-0.3,0) {$2$};
			\draw (1.5,3) node[circle,fill,blue,inner sep=1.5pt] {};
			\node[inner sep=0,anchor=west,text width=3.3cm] (note1) at (1.2,3) {$1$};

			\draw (3,1.5) node[circle,fill,inner sep=1pt] {};
			\node[inner sep=0,anchor=west,text width=3.3cm] (note1) at (3.1,1.5) {$2$};
			\draw (3,0) node[circle,fill,inner sep=1pt] {};
			\node[inner sep=0,anchor=west,text width=3.3cm] (note1) at (3.1,0) {$5$};
			\draw (1.5,-1) node[circle,fill,blue,inner sep=1.5pt] {};
			\node[inner sep=0,anchor=west,text width=3.3cm] (note1) at (1.2,-1.1) {$1$};
			\draw (1,0.3) node[circle,fill,red,inner sep=1.5pt] {};
			\node[inner sep=0,anchor=west,text width=3.3cm] (note1) at (0.8,0.1) {$6$};
			\draw (2,0.3) node[circle,fill,green,inner sep=1.5pt] {};
			\node[inner sep=0,anchor=west,text width=3.3cm] (note1) at (2,0.09) {$4$};
			\draw (1.5,1.7) node[circle,fill,inner sep=1pt] {};
			\node[inner sep=0,anchor=west,text width=3.3cm] (note1) at (1.2,1.9) {$3$};
			\node[inner sep=0,anchor=west,text width=3.3cm] (note1) at (1.3,-1.5) {$(a)$};

			\end{tikzpicture}
			\hspace{-2cm}
			\quad

			\begin{tikzpicture}
			\draw[fill=red] (0,0)--(0,1.5)--(1,0.3);
			\draw[fill=black!20] (1,0.3)--(2,0.3)--(1.5,1.7);
			\draw[fill=black!20] (2,0.3)--(3,0)--(3,1.5);
			\draw[fill=black!20] (2,0.3)--(3,1.5)--(1.5,1.7);
			\draw[fill=black!20] (1.5,1.7)--(3,1.5)--(1.5,3);
			\draw[fill=black!20] (0,1.5)--(1.5,1.7)--(1.5,3);
			\draw[fill=black!20] (0,1.5)--(1,0.3)--(1.5,1.7);
			\draw[fill=black!20] (0,0)--(1,0.3)--(1.5,-1);
			\draw[fill=black!20] (2,0.3)--(1,0.3)--(1.5,-1);
			\draw[fill=black!20] (2,0.3)--(3,0)--(1.5,-1);
			\draw (0,0)--(0,1.5);
			\draw (0,1.5)--(1.5,3);
			\draw (1.5,3)--(3,1.5);
			\draw (3,1.5) -- (3,0);
			\draw (3,0) -- (1.5,-1);
			\draw (1.5,-1) -- (0,0);
			\draw (1.5,3) -- (1.5,1.7);
			\draw (1.5,1.7) -- (0,1.5);
			\draw (1.5,1.7) -- (3,1.5);
			\draw (1,0.3) -- (0,0);
			\draw (1,0.3) -- (0,1.5);
			\draw (1,0.3) -- (1.5,1.7);
			\draw (1,0.3) -- (2,0.3);
			\draw (2,0.3) -- (1.5,1.7);
			\draw [thick,green](2,0.3) -- (3,1.5);
			\draw (2,0.3) -- (3,0);
			\draw (2,0.3) -- (1.5,-1);
			\draw (1,0.3) -- (1.5,-1);
			\draw[thick,->] (2,0.3) -- (1.8,-0.2);
			\draw[thick,->] (1,0.3) -- (1.2,-0.2);

			\draw (0,1.5) node[circle,fill,inner sep=1pt] {};
			\node[inner sep=0,anchor=west,text width=3.3cm] (note1) at (-0.3,1.5) {$5$};
			\draw (0,0) node[circle,fill,inner sep=1pt] {};
			\node[inner sep=0,anchor=west,text width=3.3cm] (note1) at (-0.3,0) {$2$};
			\draw (1.5,3) node[circle,fill,blue,inner sep=1.5pt] {};
			\node[inner sep=0,anchor=west,text width=3.3cm] (note1) at (1.2,3) {$1$};

			\draw (3,1.5) node[circle,fill,inner sep=1pt] {};
			\node[inner sep=0,anchor=west,text width=3.3cm] (note1) at (3.1,1.5) {$2$};
			\draw (3,0) node[circle,fill,inner sep=1pt] {};
			\node[inner sep=0,anchor=west,text width=3.3cm] (note1) at (3.1,0) {$5$};
			\draw (1.5,-1) node[circle,fill,blue,inner sep=1.5pt] {};
			\node[inner sep=0,anchor=west,text width=3.3cm] (note1) at (1.2,-1.1) {$1$};
			\draw (1,0.3) node[circle,fill,inner sep=1pt] {};
			\node[inner sep=0,anchor=west,text width=3.3cm] (note1) at (0.8,0.1) {$6$};
			\draw (2,0.3) node[circle,fill,inner sep=1pt] {};
			\node[inner sep=0,anchor=west,text width=3.3cm] (note1) at (2,0.09) {$4$};
			\draw (1.5,1.7) node[circle,fill,inner sep=1pt] {};
			\node[inner sep=0,anchor=west,text width=3.3cm] (note1) at (1.2,1.9) {$3$};
			\draw[thick,->] (3,0.75) -- (2.6,0.65);
			\draw[thick,->] (2.5,0.17) -- (2.28,-0.2);

			\draw[thick,->] (1.5,0.3) -- (1.5,-0.2);
			\draw[thick,->] (0.5,0.17) -- (0.7,-0.2);
			\draw[thick,->] (0,0) -- (0.42,-0.3);
			\draw[thick,->] (3,0) -- (2.59,-0.3);
			\draw[thick,->] (0.55,0.85) -- (0.9,1);
			\draw[thick,->] (1.27,1.08) -- (1.6,0.85);
			\draw[thick,->] (1.8,0.9) -- (2.1,1.1);
			\draw[thick,->] (0,1.5) -- (0.37,1.85);
			\draw[thick,->] (0.75,1.6) -- (0.95,1.95);
			\draw[thick,->] (1.5,1.7) -- (1.5,2.2);
			\draw[thick,->] (2.2,1.6) -- (2.1,2);
			\draw[thick,->] (3,1.5) -- (2.6,1.9);
			\node[inner sep=0,anchor=west,text width=3.3cm] (note1) at (1.3,-1.5) {$(b)$};
			\end{tikzpicture}
		}%
		\caption{$(a)$ An injective scalar function $f$ defined on the vertices of the real projective space $\mathbb{RP}^2$. $(b)$ An RP discrete gradient vector field with respect to $f$ on $\mathbb{RP}^2$. }
		\label{fig:1-1 correspondence}
	\end{figure}

\begin{remark}\label{rem:arbitraryDIM}
  It is worth to be noticed that the above corollary ensures also that a PL critical point $v$ of $f$ and the discrete critical $i$-simplices in correspondence with $v$ are closely located. More precisely, given any discrete critical $i$-simplex $\sigma$ in correspondence with $v$, we have that $\sigma$ belongs to $\sta(v)$ or, equivalently, that $v$ is a vertex of $\sigma$.
\end{remark}

An example for Corollary \ref{cor:arbitraryDIM} (also for Corollary \ref{cor:lowDIM}) is given in Figure \ref{fig:1-1 correspondence}.  We consider the real projective space $\mathbb{RP}^2$ endowed with an injective scalar function $f$ defined on its vertices depicted in Figure \ref{fig:1-1 correspondence}$(a)$.  In Figure \ref{fig:1-1 correspondence}$(a)$, the vertices marked in blue, green and red represent the minimum, saddle, and maximum critical points of the PL Morse function $f$, respectively.  In Figure \ref{fig:1-1 correspondence}$(b)$, we depict the RP discrete gradient vector field with respect to $f$ on $\mathbb{RP}^2$, which is obtained according to the algorithm given in Theorem \ref{theorem RP}, and we highlight the minimum, saddle and maximum discrete critical simplices in blue, green and red, respectively.  We see that PL minimum, saddle and maximum critical points are in $1$-to-$1$ correspondence with the minimum, saddle and maximum discrete critical simplices, respectively.  Moreover, as claimed by Remark \ref{rem:arbitraryDIM}, each discrete critical simplex belongs to the star of the corresponding PL critical point.

\subsection{Construction of RP discrete gradient vector fields}\label{subsec:existence}

In this subsection, we prove that, for combinatiorial manifolds of dimension $d\leq 3$, the existence of an RP discrete gradient vector field (and, consequently, a correspondence between PL and discrete critical sets) is always ensured.


\begin{lemma}\label{freeface}
Let $\Sigma$ be a simplicial complex of dimension $2$ such that $|\Sigma|\subsetneq S^2$.  Then, there exists a triangle in $\Sigma$ admitting a free face (i.e., an edge belonging to exactly one triangle of $\Sigma$).
\end{lemma}

\begin{proof}
Let us suppose that none of the triangles in  $\Sigma$ admits a free face.  Then, every face of each triangle in $\Sigma$ belongs to exactly two triangles, and the collection of all triangles forms at least one $2$-cycle in $\Sigma$ which is not a boundary.  Thus, $\dim(H_2(\Sigma)) \geq 1$.  Let $\Sigma'$ be a simplicial complex such that $|\Sigma'|= \overline{S^2 \setminus |\Sigma|}$.  Since $S^2=|\Sigma'| \cup |\Sigma|$,  $|\Sigma|\subsetneq S^2$ and $\Sigma$ is a simplicial complex of dimension $2$, $\Sigma'$ is a $2$-dimensional simplicial complex such that  $|\Sigma'|\subsetneq S^2$.  Similarly to the case of $\Sigma$, $\dim(H_2(\Sigma')) \geq 1$.  Thus, we get  $\dim(H_{2}(|\Sigma'|)\oplus H_{2}(|\Sigma|))\geq 2.$
Thanks to Mayer-Vietoris sequence for homology, we have the following long exact sequence
$$
0 \rightarrow  H_2(|\Sigma'| \cap |\Sigma|)\rightarrow H_{2}(|\Sigma'|)\oplus H_{2}(|\Sigma|)\xrightarrow{\phi} H_{2}(S^2) \rightarrow \cdots
$$
Moreover, since $\dim(|\Sigma'| \cap |\Sigma|)\leq 1$, $H_2(|\Sigma'| \cap |\Sigma|)=0$.
Thus, we get that the map $\phi: H_{2}(|\Sigma'|)\oplus H_{2}(|\Sigma|) \rightarrow H_{2}(S^2)$ is injective which is clearly not possible since $\dim(H_{2}(|\Sigma'|)\oplus H_{2}(|\Sigma|))\geq 2$ and $\dim(H_{2}(S^2))=1.$
\qed \end{proof}

\begin{lemma}\label{pgvf}
Let $\Sigma$ be a simplicial complex such that $|\Sigma|\subseteq S^2$.  Then, $\Sigma$ admits a perfect discrete gradient vector field.
\end{lemma}

\begin{proof}
We can assume, without loss of generality, that $\Sigma$ is connected. If it is not the case, the lemma can be proved by considering each component separately.
If $\dim(\Sigma)\leq1$, then $\Sigma$ admits a perfect discrete gradient vector field by \cite{Lewiner2003}.  If $\dim(\Sigma)=2$, then there are two cases: either $|\Sigma|= S^2$ or $|\Sigma|\subsetneq S^2$.
\begin{enumerate}
  \item If $|\Sigma|= S^2$, then $\Sigma$ admits a perfect discrete gradient vector field by \cite{Lewiner2003}.
  \item If $|\Sigma| \subsetneq S^2$, then $\Sigma$ admits a triangle $\tau_1$ with a free face $\sigma_1$ by Lemma \ref{freeface}.  Let $\Sigma^1= \Sigma \setminus \{\sigma_1, \tau_1\}$.
  If $\dim(\Sigma^1) \leq 1$, then $\Sigma^1$ admits a perfect discrete gradient vector field $V^1$ by \cite{Lewiner2003}.  Since $\Sigma=\Sigma^1 \cup \{\sigma_1, \tau_1\}$ and removal of the pair $\{\sigma_1, \tau_1\}$ is a collapse, then $\Sigma$ and $\Sigma^1$ have isomorphic homology groups.
  Thus, $V=V^1\cup \{(\sigma_1, \tau_1)\}$ is a perfect discrete gradient vector field on $\Sigma$ with $m_i(V)=m_i(V^1)$ for all $i\geq0$.
  If $\dim(\Sigma^1)=2$, then we remove pairs of simplices $\{\sigma_j, \tau_j\}$ successively from $\Sigma^1$ (which is possible by Lemma \ref{freeface}) up to getting a subcomplex  $\Sigma^{n}$ of $\Sigma^1$ such that $\dim(\Sigma^n) \leq 1$ and $\Sigma^1=\Sigma^{n} \cup \{\sigma_2, \tau_2, \ldots, \sigma_n, \tau_n\}$ where $\{\sigma_j, \tau_j\}$ is the edge-triangle pair removed at the $j^{th}$ step of the successive operation.
  Since each removal of a pair of simplices $\{\sigma_j, \tau_j\}$ is a collapse, $\Sigma^1$ and $\Sigma^{n}$ have isomorphic homology groups.  Let $V^n$ be a perfect discrete gradient vector field on $\Sigma^{n}$.  Then, $V^1=V^n \cup \{(\sigma_2, \tau_2), \ldots, (\sigma_n, \tau_n)\}$ is a perfect discrete gradient vector field on $\Sigma^1$ with $m_i(V^1)=m_i(V^n)$ for all $i\geq0$.  Since $\Sigma=\Sigma^1 \cup \{\sigma_1, \tau_1\}$, $\Sigma$ and $\Sigma^1$ have isomorphic homology groups.  So, $V=V^1\cup \{(\sigma_1, \tau_1)\}$ is a perfect discrete gradient vector field on $\Sigma$ with $m_i(V)=m_i(V^1)$ for all $i$. \qed
\end{enumerate}
\end{proof}

\begin{theorem}\label{theorem RP}
  Let $\Sigma$ be a combinatorial $d$-manifold with $d\leq3$ and let $f:\Sigma_0 \to \mathbb{R}$ be an injective function.  Then, there exists a discrete gradient vector field $V$ on $\Sigma$ that is RP with respect to $f$.
\end{theorem}
We give a proof of Theorem \ref{theorem RP} for $d=3$.  The proof for $d\leq2$ is similar to the case for $d=3$.
\begin{proof}
  Since $f$ is injective,  $\sta^{-}(v)\cap \sta^{-}(v')=\emptyset$ for any $v\neq v' \in \Sigma_0$.  Hence, $\Sigma$ can be constructed as a disjoint union of lower stars, that is,  $\Sigma= \coprod_{v\in \Sigma_0}\sta^{-}(v)$.  Since $\Sigma$ is a combinatorial $3$-manifold,  $|\link^-(v)|\subseteq S^2$.  If $\link^{-}(v)=\emptyset$, then $\sta^{-}(v)=\{v\}$ and any discrete gradient vector field $W$ on $\sta^{-}(v)$ admits exactly one critical simplex, which is the $0$-simplex $v$.  So, $W$ is a perfect discrete gradient vector field on $\sta^{-}(v)$.  Assume that $\link^{-}(v)\neq \emptyset$.  By Lemma \ref{pgvf}, $\link^-(v)$ admits a perfect discrete gradient vector field $W$. By Lemma \ref{equation}, 	$H_{i+1}(\overline{\sta^{-}(v)}, \link^{-}(v))\cong \tilde{H}_{i}(\link^{-}(v)) \text{ \ for \ } i\geq0.$
  Moreover, because $\tilde{H}_{0}(\overline{\sta^{-}(v)})=0$, the long exact sequence used in Lemma \ref{equation} implies that
  \begin{eqnarray*}
    H_{0}(\overline{\sta^{-}(v)}, \link^{-}(v))\cong	\tilde{H}_0(\overline{\sta^{-}(v)}, \link^{-}(v))=0.
  \end{eqnarray*}
  So,
  \begin{align*}
  &\beta_{i+1}(\overline{\sta^{-}(v)}, \link^{-}(v))= \beta_{i}( \link^{-}(v)) \text{ for  $i>0$,}\\
  &\beta_{1}(\overline{\sta^{-}(v)}, \link^{-}(v))=\beta_{0}( \link^{-}(v))-1,\\ &\beta_{0}(\overline{\sta^{-}(v)}, \link^{-}(v))=0.
\end{align*}
  Since $W$ is a perfect discrete gradient vector field, 	$m_i(W) = \beta_i(\link^{-}(v)) \text{ for } i=0,1,2.$

  Let $n_i$ denote the number of $i$-simplices in $\link^{-}(v)$ for $i=0,1,2$, and $n'_i$ denote the number of $i$-simplices in $\sta^{-}(v)$ for $i=0,1,2,3$.  Since there is a 1-to-1 correspondence between the $i$-simplices in $\link^{-}(v)$ and $(i+1)$-simplices in $\sta^{-}(v)$, $n'_0=1, n'_1=n_0, n'_2=n_1 $ and $n'_3=n_2$.

  Now, we construct a vector field $W'$ on $\sta^{-}(v)$ as follows.
  \begin{itemize}
    \item If $(\alpha, \beta)\in W$, then we set $(v\alpha, v\beta)\in W'$, where, given a simplex $\sigma$, $v\sigma$ denotes the simplex spanned by $v$ and the vertices of $\sigma$.
    \item If $\gamma$ is a discrete critical $i$-simplex of $W$ with $i>0$, then we set $v\gamma$ as critical for $W'$.
    \item If $\gamma_1, \gamma_2, \dots, \gamma_{m_0(W)}$ are the discrete critical $0$-simplices of $W$, then we set $(v, v\gamma_1)\in W'$ and $v\gamma_2, \dots, v\gamma_{m_0(W)}$ as critical for $W'$.
  \end{itemize}
  Since $W$ is a discrete gradient vector field, it does not admit any closed path. By construction, $W'$ does not admit any closed path. Hence, $W'$ is a discrete gradient vector field on $\sta^{-}(v)$. The numbers of discrete critical simplices of $W'$ are as follows:
  \begin{align*}
  m_{0}(W')&=0=\beta_0((\overline{\sta^{-}(v)}, \link^{-}(v))),\\
  m_{1}(W')&=m_{0}(W)-1=\beta_{0}( \link^{-}(v))-1\\&=\beta_1((\overline{\sta^{-}(v)}, \link^{-}(v))),\\
  m_{i}(W')&=m_{i-1}(W)=\beta_{i-1}( \link^{-}(v))\\&=\beta_i((\overline{\sta^{-}(v)}, \link^{-}(v))) \text{\ for \ } i=2,3.
\end{align*}
  Let $V$ be the collection of all discrete gradient vector fields $W'$ on $\sta^{-}(v)$ for each $v$. Since $\Sigma= \coprod_{v\in \Sigma_0}\sta^{-}(v)$, $V$ is a discrete gradient vector field on $\Sigma$ whose restriction to each $\sta^{-}(v)$ is $W'$.  Let $f(v)=l$.  Since $\Sigma^{l}\setminus\Sigma^{l'}= \sta^{-}(v)$ and  $\Sigma^{l'}=\Sigma^{l} \setminus \sta^{-}(v)$,
  in accordance with Equation (\ref{brehm}), we get $m^l_i(V) = m_i(W')=\beta_i((\overline{\sta^{-}(v)}, \link^{-}(v)))=\beta_{i}(\Sigma^l, \Sigma^{l'})$.
  Thus, $V$ is a discrete gradient vector field on $\Sigma$ that is RP with respect to $f$.
\qed \end{proof}
An example for the algorithm given in the proof of the Theorem \ref{theorem RP} is given in Figure \ref{fig:correspondence}.  In Figure \ref{fig:correspondence}$(a)$, we depict a portion of an RP discrete gradient vector field given on a combinatorial $2$-manifold $\Sigma$ and obtained according to the Theorem \ref{theorem RP}. In Figure \ref{fig:correspondence}$(b)$, we illustrate how to obtain the discrete gradient vector field $W'$ on the lower star, $\sta^{-}(v)$, of the vertex $v:=f^{-1}(5)$, pictorially.  A discrete gradient vector field $W$ (highlighted in blue) on $\link^{-}(v)$ consists of the pair $(w, uw)$ and the discrete critical $0$-simplices $u$ and $z$, where by $u$, $w$, $z$ we denote $f^{-1}(1)$, $f^{-1}(2)$, $f^{-1}(3)$, respectively. To obtain $W'$ from $W$, we pair the $1$-simplex $vw$ with the $2$-simplex $uvw$, the $0$-simplex $v$ with the $1$-simplex $uv$ and we set the $1$-simplex $vz$ as critical. In Figure \ref{fig:correspondence}$(c)$, we depict a portion of an arbitrary discrete gradient vector field on $\Sigma$ and we show that the discrete gradient vector field is not RP because $v$ is a discrete critical $0$-simplex of $\Sigma^5\setminus \Sigma^3$, that is $m^5_0=1$, but $\beta_0(\Sigma^5, \Sigma^{3})=0$.

By combining Theorem \ref{theorem RP} with Corollary \ref{cor:arbitraryDIM}, we get the following result.
\begin{corollary}\label{cor:lowDIM}
  Let $\Sigma$ be a combinatorial $d$-manifold with $d\leq3$ and let $f:\Sigma_0 \to \mathbb{R}$ be an injective function.
  Then, there exists a discrete gradient vector field $V$ on $\Sigma$ (RP w.r.t. $f$) such that there is a 1-to-$k_i$ correspondence between PL critical points of index $i$ and multiplicity $k_i$ of $f$  and discrete critical $i$-simplices $\sigma$ of $V$ such that  $f_{max}(\sigma)=f(v)$. If $f$ is PL Morse, then the correspondence is bijective.
\end{corollary}


\section{Conclusions and future developments}\label{sec:conclusions}

In this paper, we have studied the link between PL critical points and discrete critical simplices.
Our investigation has identified a condition (that can always be met in the case of domains of dimension lower than or equal to 3) under which a correspondence between the two critical sets exists. Moreover, if the function is PL Morse, the retrieved correspondence is a bijection.
Our results can serve as a theoretical  guarantee to practitioners of Morse theory in applications about the gains and the losses determined by chosing one or the other of the PL and discrete Morse theory. In practice, we ensure that in low dimensions there is no loss of information switching between the two in terms of either size or localization of critical sets.

A number of interesting questions remain open, which deserve further attention from us in the near future:
\begin{itemize}
\item In first place, we plan to understand the relationships among the PL and combinatorial notions of some interesting cellular decompositions associated to a Morse function, such as the ones given by the ascending and descending regions (also known as unstable and stable regions) and the Morse-Smale complexes. To the best of our knowledge, precise definitions of such regions are missing for either one or the other of the two theories, and similarly for some results such as the Quadrangle  Lemma \cite{Edelsbrunner03}.
\item In second place, it would be interesting to define precisely a procedure that allows us to unfold any non-simple PL singularity to obtain a Morse function. Studying such unfoldings for the PL and combinatorial setting should be very useful in applications as an approximation of the well known smooth bifurcation theory.
\item It is known that steepest descent PL flows can merge and fork even on a surface. In discrete Morse theory, their analogues, that are $V$-paths, cannot be involved in flows which merge and fork. Can this be used to simplify algorithms in PL Morse theory?
\item It is known that discrete Morse functions mirror the behaviour of  smooth Morse functions in terms of Morse vectors, that are arrays whose $i^{th}$ element is given by the number of critical points of index $i$ \cite{Benedetti12}. Indeed,  up to barycentric subdivisions, if a smooth manifold admits a Morse vector $\mathbf c$ then there is a discrete Morse function with the same Morse vector. This is true in any dimension. Can we leverage this result to say something more about relative perfectness beside dimension 3?
\end{itemize}

\section*{Acknowledgments}
The first author acknowledges the support from the Italian MIUR Award ``Dipartimento di Eccellenza 2018-2022" - CUP: E11G18000350001 and the SmartData@PoliTO center for Big Data and Machine Learning.

\bibliographystyle{elsarticle-num}
\bibliography{PLdiscreteBiblio}

\end{document}